\documentclass[conference]{IEEEtran}
\IEEEoverridecommandlockouts
\usepackage{cite}
\usepackage{amsmath,amssymb,amsfonts}
\usepackage{amsthm}
\usepackage{algorithm}
\usepackage{algorithmic}
\usepackage{graphicx}
\usepackage{textcomp}
\usepackage{xcolor}
\usepackage{multirow} 
\usepackage{booktabs} 
\def\BibTeX{{\rm B\kern-.05em{\sc i\kern-.025em b}\kern-.08em
    T\kern-.1667em\lower.7ex\hbox{E}\kern-.125emX}}
    
\usepackage[unicode=true, bookmarks=true,bookmarksnumbered=true,bookmarksopen=true,bookmarksopenlevel=1, breaklinks=false,pdfborder={0 0 0},pdfborderstyle={},backref=false,colorlinks=true]{hyperref}

\ifCLASSOPTIONcompsoc
\usepackage[caption=false, font=normalsize, labelfont=sf, textfont=sf]{subfig} 
\else
\usepackage[caption=false, font=footnotesize]{subfig}
\usepackage{geometry}
\newtheorem{theorem}{Theorem}
\newtheorem{proposition}{Proposition}

\newtheorem{remark}{Remark}

\makeatother

\begin{document}
\title
{
A Peaceman-Rachford Splitting Approach with Deep Equilibrium Network for Channel Estimation 
\thanks{This work was partially supported by National Natural Science
Foundation of China (12271289 and 12025104).}
}

\author{
    \IEEEauthorblockN{Dingli Yuan}
    \IEEEauthorblockA{
        \textit{Department of Mathematical Sciences} \\
        \textit{Tsinghua University}\\
        Beijing, China \\
        Email: yuandl22@mails.tsinghua.edu.cn}
    ~\\
    \and
    \IEEEauthorblockN{Shitong Wu*}
    \IEEEauthorblockA{
        \textit{Department of Mathematical Sciences} \\
        \textit{Tsinghua University}\\
        Beijing, China \\
        Email: wust20@mails.tsinghua.edu.cn}
    *Corresponding author
    ~\\
    \and
    \IEEEauthorblockN{Haoran Tang}
    \IEEEauthorblockA{
        \textit{Department of Mathematical Sciences} \\
        \textit{Tsinghua University}\\
        Beijing, China \\
        Email: thr22@mails.tsinghua.edu.cn}
    ~\\
    \and
    \IEEEauthorblockN{Lu Yang}
    \IEEEauthorblockA{
        \textit{Wireless Technology Lab} \\
        \textit{Central Research Institute, 2012 Labs, Huawei Tech. Co. Ltd.}\\
        China \\
        Email: yanglu87@huawei.com}
    ~\\
    \and
    \IEEEauthorblockN{Chenghui Peng}
    \IEEEauthorblockA{
        \textit{Wireless Technology Lab} \\
        \textit{Central Research Institute, 2012 Labs, Huawei Tech. Co. Ltd.}\\
        China \\
        Email: pengchenghui@huawei.com}
}

\maketitle

\begin{abstract}
Multiple-input multiple-output (MIMO) is pivotal for wireless systems, yet its high-dimensional, stochastic channel poses significant challenges for accurate estimation, highlighting the critical need for robust estimation techniques. In this paper, we introduce a novel channel estimation method for the MIMO system. The main idea is to construct a fixed-point equation for channel estimation, which can be implemented into the deep equilibrium (DEQ) model with a fixed network.
Specifically, the Peaceman-Rachford (PR) splitting method is applied to the dual form of the regularized minimization problem to construct fixed-point equation with non-expansive property.
Then, the fixed-point equation is implemented into the DEQ model with a fixed layer, leveraging its advantage of the low training complexity.
Moreover, we provide a rigorous theoretical analysis, demonstrating the convergence and optimality of our approach. Additionally, simulations of hybrid far- and near-field channels demonstrate that our approach yields favorable results, indicating its ability to advance channel estimation in MIMO system. 
\end{abstract}
\begin{IEEEkeywords}
channel estimation, deep learning, fixed-point, Peaceman-Rachford splitting method, deep equilibrium model, hybrid-field.
\end{IEEEkeywords}

\section{Introduction}

Multiple-input multiple-output (MIMO) technology is pivotal in 6G wireless systems \cite{sarieddeen2021overview}. 
Recognizing its paramount importance, the channel estimation algorithms should have high adaptability and  performance for different types of channels in the MIMO scenario. 
For example, it need to be feasible for the hybrid far- and near-field channels \cite{wei2021channel_hybrid,yu2022hybrid}, etc. 
Hence, there is a strong motivation for improvements in channel estimation algorithms with broad applicability and high accuracy.

Traditional channel estimation schemes with closed-form expressions, such as least square (LS) and minimum mean-squared error (MMSE), have limitations in practical applications: the LS significantly lacks in estimation performance, while the MMSE's complexity is prohibitively high \cite{tuchler2002minimum}. 
Guided by the pursuit of higher estimation performance and lower computation complexity, numerous iterative algorithms have been introduced. 
One approach involves enhancing estimation performance by adding regularization terms after least squares enabling the use of optimization methods, such as proximal gradient descent (PGD) or alternating direction method of multipliers (ADMM) \cite{boyd2011distributed}. 
Another approach employs probabilistic models with certain types of measurement matrices, such as Gaussian or row-orthogonal matrices, which utilizes prior information assumed theoretically to minimize the mean square error.
Classical algorithms include approximate message passing (AMP) \cite{donoho2009message}, Orthogonal AMP (OAMP) algorithms \cite{ma2017orthogonal}, and their derivative algorithms \cite{rangan2019vector}. 
Despite differing approach, these algorithms share similar iteration update rules, composed of a linear estimator (LE) and a non-linear estimator (NLE) \cite{yu2023adaptive,gilton2021DEM_inverse}. 
The LE is explicit with low complexity. 
The bottleneck of a high performance estimator design lies in the NLE \cite{yu2023adaptive}, since its design requires prior knowledge of the channel, which is of high complexity and difficult to acquire. 
A heuristic way is to replace NLE with neural network (NN) \cite{he2020oamp_net_v2}, which has advantage in capturing data features. 
Additionally, deep neural networks, which have demonstrated strong denoising capabilities \cite{liu2022denoising}, are theoretically well-suited to replace the NLE.
Therefore, this approach by replacing the NLE with network motivates the thriving of model-driven channel estimators and demonstrates significant performance in applications \cite{he20modeldriven}. 

However, directly replacing NLE with deep learning-based estimators, particularly those employing deep unfolding networks, presents some systematic challenges, as they are constructed by unfolding a classical iterative algorithm across multiple layers \cite{zhang2018ista_net}.
First, its reliability lacks theoretical guarantees, as truncating the algorithm into $T$ layers disrupts the convergence of a classical iterative algorithm. 
Second, it constructs different NNs for each layer of iteration, which entails the intermediate states and gradients each layer, resulting in substantial memory and computational overhead during training and exhibiting poor scalability for deeper number of truncated layers.
Considering these issues, a reconsideration of the feasibility of the deep unfolding framework is imperative.
Therefore, \cite{bai2019DEM} proposed the deep equilibrium (DEQ) model, which utilizes a single fixed layer to represent a fixed-point equation, ensuring that the output converges to the fixed point as the number of iterations approaches infinity.
Moreover, since the DEQ model uses a fixed layer for each iteration, its training complexity, compared to the deep unfolding network with complexity $O(T)$, can be reduced to $O(1)$ \cite{bai2019DEM}. 
Recent works have consider the DEQ model with its application in inverse problem \cite{gilton2021DEM_inverse}, image demonising \cite{liu2022denoising}, video snapshot compressive imaging \cite{zhao2022deep_snapshot} and hybrid far-and near-field channel estimation \cite{yu2022hybrid,yu2023adaptive}. 
In these work, they replace the NLE with a fixed-point network and trained the model by controlling the Lipschitz constant to ensure a linear convergence rate \cite{pabbaraju2020lips}. 
Specifically, reference \cite{yu2022hybrid} demonstrates the powerful performance in hybrid-field channel estimation by learning OAMP algorithm with a fixed-point network.
However, the OAMP framework requires certain assumptions about the measurement matrix \cite{ma2017orthogonal,yu2022hybrid}, leading to the adjustment and check of the network's Lipschitz constant in each iteration to ensure the contraction property of the DEQ model, which not only increases the complexity but also reduces the reliability of the solutions.
Hence, to fully utilize the DEQ model with a fixed-point network, a carefully designed fixed-point equation needs to be constructed, ensuring that the system possesses a controllable Lipschitz constant \cite{pabbaraju2020lips}.

In this paper, we derive a fixed-point formulation for the application of the DEQ model with a controllable Lipschitz constant in channel estimation, offering convergence guarantees and reduced training complexity.
Unlike traditional methods that model channel estimation problem through minimization from the prime perspective, our approach explores the problem from the dual formulation to leverage the inherent convexity of the dual perspective. 
Building on this, we apply the Peaceman-Rachford (PR) splitting method to derive the fixed-point equation on the dual problem, ensuring a non-expansive property of the combined operator \footnote{non-expansive means that the Lipschitz constant is no larger than one \cite{Bauschke2011monotone_book}} and robust convergence properties of the fixed-point iteration. 
To simplify the calculation, we derive an alternative algorithm for the Peaceman-Rachford splitting method in the dual form, and unfold one iteration of the proposed algorithm into a fixed layer by replacing the nonlinear term with a specially-trained neural network. 
Since this approach leverages the DEQ framework for analytical computation of the fixed-point iteration, it is named as Peaceman-Rachford splitting method implemented with the Deep Equilibrium Network (PR-DEN). 
Additionally, we present a comprehensive theoretical analysis providing rigorous proofs of both convergence and optimality.
Extensive simulations, including the scenario with MIMO systems featuring hybrid far- and near-field channels, validate the robust performance of our methodology, which underscore the potential of our proposed approach to significantly advance channel estimation in MIMO systems.

\section{System Model and Problem Formulation}

In this work, we consider the uplink channel estimation for MIMO system as shown in Fig. \ref{model_system}. 
The base station (BS) is equipped with $S$ uniform planar arrays (UPAs) arranged in a $\sqrt{S} \times \sqrt{S}$ configuration and has a total of $N$ antennas.
Specifically, the separation between adjacent UPAs is denoted as $d_\text{UPA}$. 
In each UPA, the antennas are arranged in a $\sqrt{N/S}\times\sqrt{N/S}$ configuration ($\sqrt{N/S} \in \mathbb{Z}^+$) with the antenna spacing of $d$. 
Additionally, the BS is employed with a hybrid precoding architecture \cite{alkhateeb2015limited}, featuring $N_{\text{RF}}$ radio-frequency (RF) chains and serving $K$ single-antenna users, where $N_{\text{RF}}\ll N$ and $K \leq N_{\text{RF}}$. The antennas in each UPA share the same RF chain.
For the uplink channel estimation, we assume the $K$ users transmit mutual orthogonal pilot sequences to the BS \cite{tse2005fundamentals}, then channel estimation for each user
is independent. Therefore, without loss of generality, we consider the scenario where an arbitrary user is involved.
\begin{figure}[b]
\centering
\vspace{-0.98em} 
\includegraphics[width=1\linewidth]{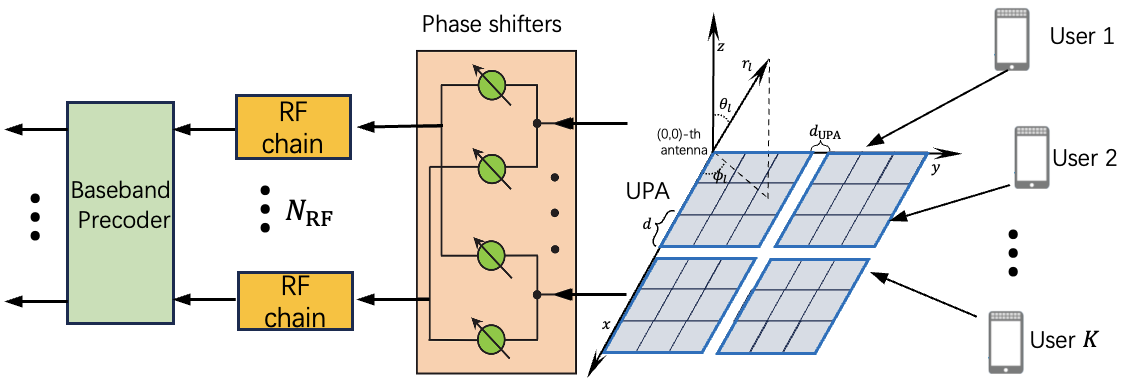}
\caption{System model of the uplink channel estimation.}
\label{model_system}
\end{figure}

The user transmits known pilot signals to the BS for $P$ time slots, and we denote $s_{p}\in\mathbb{C}$ as the transmit pilot in the time slot $p$. 
Then, the received pilot $\boldsymbol{y}_{p} \in \mathbb{C}^{N_{\mathrm{RF}} \times 1}$ is
$$
\boldsymbol{y}_{p}=\mathbf{W}_p\mathbf{A}_{p}\mathbf{F} \boldsymbol{h} s_{ p}+\mathbf{W}_p\mathbf{A}_{p} \boldsymbol{n}_{p},
$$
where $\mathbf{F} \in \mathbb{C}^{N\times N}$ denotes the Fourier transform matrix transforming the channel $\boldsymbol{h}\in\mathbb{C}^{N\times 1}$ into its angular-domain, $\mathbf{A}_{p} \in \mathbb{C}^{N_{\mathrm{RF}} \times N}$ denotes the analog combining matrix and satisfies the constant modulus constraint, \textit{i.e.}, $\left| \mathbf{A}_{p}(i, j)\right|=\frac{1}{\sqrt{N}}$ \cite{he2018deep}, $\mathbf{W}_p \in \mathbb{C}^{N_{\mathrm{RF}}\times N_{\mathrm{RF}}}$ denotes the digital combining matrix, and $\boldsymbol{n}_{p} \in \mathbb{C}^{N \times 1}$ denotes the Gaussian complex noise following the distribution $\mathcal{CN}(0, \sigma^{2} \mathbf{I}_{N})$.
Denote the overall received pilot sequence $\boldsymbol{y}=\left[\boldsymbol{y}_{1}^{T}, \dots, \boldsymbol{y}_{P}^{T}\right]^{T}$ and set $s_q$ as $1$, then 
\begin{equation}
 \label{complex_form}
\boldsymbol{y}=\mathbf{A}\boldsymbol{h}+\boldsymbol{n},
\end{equation}
where $\boldsymbol{n}=\left[\boldsymbol{n}_{1}^{T} \mathbf{A}_{1}^{T}\mathbf{W}_{1}^{T}, \dots, \boldsymbol{n}_{P}^{T} \mathbf{A}_{P}^{T}\mathbf{W}_{P}^{T}\right]^{T} \in \mathbb{C}^{ N_{\mathrm{RF}}P \times 1}$ denotes the noise and $\mathbf{A}=\left[\mathbf{F}^T\mathbf{A}_{1}^{T}\mathbf{W}_{1}^{T},\dots, \mathbf{F}^T\mathbf{A}_{P}^{T} \mathbf{W}_{P}^{T}\right]^{T} \in \mathbb{C}^{N_{\mathrm{RF}}P \times N} $ denotes the overall analog combining matrix.

In the following, we adopt a real form where the matrices and vectors $\mathbf{A}$, $\boldsymbol{y}$, $\boldsymbol{h}$ and $\boldsymbol{n}$ are substituted by: 
$$\begin{pmatrix}
\Re(\mathbf{A}) & -\Im(\mathbf{A}) \\
\Im(\mathbf{A}) & \Re(\mathbf{A})
\end{pmatrix},
\begin{pmatrix}
    \Re(\boldsymbol{y}) \\  \Im(\boldsymbol{y})
\end{pmatrix},
\begin{pmatrix}
    \Re(\boldsymbol{h}) \\  \Im(\boldsymbol{h})
\end{pmatrix},
\begin{pmatrix}
    \Re(\boldsymbol{n}) \\  \Im(\boldsymbol{n})
\end{pmatrix}, 
$$ 
where $\Im$ and $\Re$ denote the imaginary and real parts.

According to \eqref{complex_form}, the goal of the channel estimation is to recover the channel $\boldsymbol{h}$ with the knowledge of $\mathbf{A}$, $\boldsymbol{y}$ and the estimation of the noise $\boldsymbol{n}$.
Since matrix $\mathbf{A}$ may be ill-posed, directly applying LS for channel estimation does not work well \cite{boyd2018introduction}.
Hence, by introducing a regularization term \cite{beck2009fista} that leverages prior information about the channel, \eqref{complex_form} can be transformed into the following optimization problem \cite{hu2023learnable}:
\begin{equation}
    \label{eq_form_cs}
    \min_{\boldsymbol{h}} ~~ g(\boldsymbol{h})+\frac{1}{2}\|\boldsymbol{y}-\mathbf{A}\boldsymbol{h}\|_2^2, 
\end{equation}
where $g(\boldsymbol{h})$ denotes the regularization function which encompasses the prior information of the channel characteristics \cite{rangan2019vector} from a Bayesian perspective in Maximum A Posterior (MAP) estimation theory \cite{chan2016plug}.

\begin{remark}
The nonlinear regularization term in \eqref{eq_form_cs} mitigates the ill-posed nature of the channel estimation problem by effectively capturing key characteristics such as low-rank structure and sparse representations \cite{mou2022deep, li18lowrank, qi2011optimized}.
\end{remark}

\section{A Peaceman-Rachford Splitting Approach} 
\label{sec:A Peaceman-Rachford Splitting Approach with Deep Equilibrium Network}

In this section, we derive a fixed-point formulation using the PR splitting method to find the optimal solution of the dual form of \eqref{eq_form_cs}. 
We start by examining the dual form, leveraging its inherent convexity to ensure robust performance in the PR splitting method. 
Due to the intractable conjugate functions in the dual form, we derive Algorithm \ref{HPR2_alg} to address this challenge, validating its effectiveness with convergence analysis.

\subsection{Dual Problem}
As presented in the previous discussion, the regularized term $g(\boldsymbol{h})$ in \eqref{eq_form_cs} encapsulates the prior information of channel characteristics, which may not be convex.
In order to solve the problem in a more stable and robust manner, we consider the dual form \cite{rockafellar1970convex} of the above regularized problem, which is a fundamental technique in optimization \cite{boyd2004convex}.
We introduce auxiliary variables $\boldsymbol{p}$ and $\boldsymbol{q}$, constrained by $\boldsymbol{p} - \boldsymbol{q} = 0$, specifically to replace $\boldsymbol{h}$ and facilitate the derivation of the dual formulation in \eqref{eq_form_cs}.
Substituting these variables, the equivalent problem can be derived:
\begin{equation}
\label{equi_form}
\begin{aligned}
    &\min_{\boldsymbol{p},\boldsymbol{q}}& \quad g(\boldsymbol{p})+f(\boldsymbol{q}) \\
    &\text{ s.t.}& \quad\quad \boldsymbol{p}-\boldsymbol{q}=0,
\end{aligned}
\end{equation}
where $f(\boldsymbol{h}) = \frac{1}{2} \|\boldsymbol{y} - \mathbf{A}\boldsymbol{h}\|_2^2$. 
The minimization problem \eqref{equi_form} is equivalent to \eqref{eq_form_cs}, with $\boldsymbol{p}^*$, $\boldsymbol{q}^*$ the optimal solutions of \eqref{equi_form}, equating to $\boldsymbol{h}^*$ the optimal solutions of \eqref{eq_form_cs}. We will develop the subsequent content based on \eqref{equi_form}.

The Lagrangian for this problem with the Lagrange multiplier $\boldsymbol{w}$ (also known as the dual variable) is given by:
\begin{equation*}
\mathcal{L}(\boldsymbol{p}, \boldsymbol{q}; \boldsymbol{w}) = g(\boldsymbol{p}) + f(\boldsymbol{q}) + \boldsymbol{w}^\top(\boldsymbol{p} - \boldsymbol{q}).
\end{equation*}
Minimizing the Lagrangian over $\boldsymbol{p}$ and $\boldsymbol{q}$, we have: 
\begin{equation*}
\inf_{\boldsymbol{p}, \boldsymbol{q}} \mathcal{L}(\boldsymbol{p}, \boldsymbol{q}; \boldsymbol{w}) = \inf_{\boldsymbol{p}} (g(\boldsymbol{p}) + \boldsymbol{w}^\top \boldsymbol{p}) + \inf_{\boldsymbol{q}} (f(\boldsymbol{q}) - \boldsymbol{w}^\top \boldsymbol{q}).
\end{equation*}
Then, by employing the Fenchel conjugate \cite{boyd2004convex}:
\begin{equation*}
f^*(\boldsymbol{w}) = \sup_{\boldsymbol{y}} (\boldsymbol{w}^\top \boldsymbol{y} - f(\boldsymbol{y})), \quad g^*(\boldsymbol{w}) = \sup_{\boldsymbol{y}} (\boldsymbol{w}^\top \boldsymbol{y} - g(\boldsymbol{y})),
\end{equation*}
the dual form of \eqref{eq_form_cs} is expressed as:
\begin{equation}
\label{dual_pro}
    \max_{\boldsymbol{w}}~~ \{-g^*(-\boldsymbol{w})-f^*(\boldsymbol{w})\}. 
\end{equation}
Considering the inherent convexity of the Fenchel conjugate, this choice guarantees the dual form's convexity, thereby enhancing the robustness of the PR splitting method application.

\subsection{Peaceman-Rachford Splitting Algorithm}

Based on the derived dual form, we apply the PR splitting method to construct the fixed-point equation.
Using first-order conditions, the optimal solution of \eqref{dual_pro} is equivalent to solving the following rooting problem \cite{boyd2011distributed,Bauschke2011monotone_book}:
\begin{equation}
\label{rooting_pro}
    0 \in \partial g^*(-\boldsymbol{w})+ \partial f^*(\boldsymbol{w}),
\end{equation}
where $\partial$ denotes the subgradient of a convex function. The definition of subgradient can be referred to \cite{boyd2004convex}, and subgradient  can directly be proved to be a maximal monotone operator \cite{rockafellar1970convex}.

Taking $\mathbf{M}_1=\partial (g^*\circ (-\mathbf{I}))$, $\mathbf{M}_2=\partial (f^*)$, then $\mathbf{M}_1, \mathbf{M}_2$ are maximal monotone operators. 
Hence, the PR splitting method \cite{combettes2023geometry_monotone,zhang2022HPR_OT} can be directly applied to find the solution of rooting problem \eqref{rooting_pro} \cite{lions1979splitting}. 
That is, the solution of \eqref{rooting_pro} is given by 
\begin{equation*}
\boldsymbol{w} = \mathbf{J}_{\sigma \mathbf{M}_2}\boldsymbol{\eta},
\end{equation*}
where $\boldsymbol{\eta}$ satisfies the fixed-point equation \cite{lions1979splitting}:
\begin{equation}\label{PR_splitting}
\mathbf{R}_{\sigma \mathbf{M}_1}\mathbf{R}_{\sigma \mathbf{M}_2}\boldsymbol{\eta} = \boldsymbol{\eta}.
\end{equation}
Here, $\boldsymbol{\eta}$ is an intermediate variable and $\sigma \in \mathbb{R}^+$.
The term $\mathbf{J_{\sigma\mathbf{M}}} \triangleq (\mathbf{I} + \sigma\mathbf{M})^{-1}$ represents the resolvent of $\sigma\mathbf{M}$, and $\mathbf{R_{\sigma \mathbf{M}}} \triangleq 2\mathbf{J_{\sigma\mathbf{M}}} - \mathbf{I}$ is known as the reflected resolvent \cite{Bauschke2011monotone_book}. 

Specifically, the resolvent $\mathbf{J_{\sigma\partial h}}$ of the subgradient of any convex function $h$ with a given $\sigma$, denoted as $(\mathbf{I} + \sigma \partial h )^{-1}$, can be expressed using the proximal operator \cite{boyd2004convex}.

\begin{equation*}
    \begin{aligned}
    (\mathbf{I}+\sigma \partial h )^{-1}(x) = & \text{Prox}_{\sigma h}(x)\\
    \triangleq & \text{argmin}_{z}(h(z)+\frac{1}{2\sigma}\|z-x\|^2).
    \end{aligned}
\end{equation*}

\begin{remark}
\label{remark_nonexpansive}
    Since $\mathbf{M}$ is a maximal momotone operator, the reflected resolvent $\mathbf{R_{\sigma M}}$ is nonexpansive \cite[Corollary 23.11]{Bauschke2011monotone_book}, \textit{i.e.}, the Lipschitz constant of $\mathbf{R}_{\sigma \mathbf{M}}$ is not larger than one. %
    Hence, the composition $\mathbf{R}_{\sigma \mathbf{M}_1}\mathbf{R}_{\sigma \mathbf{M}_2}$ is also nonexpansive.
\end{remark}

To solve the fixed-point equation \eqref{PR_splitting}, a common approach is to use fixed-point iteration. Starting with an initial input $\boldsymbol{\eta}^0$, the following iteration is repeated until convergence:
\begin{equation}
\label{fix_eq}
\boldsymbol{\eta}^{k+1} = \mathbf{R}_{\sigma \mathbf{M}_1}\mathbf{R}_{\sigma \mathbf{M}_2}(\boldsymbol{\eta}^{k}).
\end{equation}

\subsection{Detailed Implementation and Convergence Analysis}

Since the reflected resolvent  $\mathbf{R}_{\sigma \mathbf{M}_1}$ and $\mathbf{R}_{\sigma \mathbf{M}_2}$ contain the conjecture function $f^*$ and $g^*$, it is challenging for the direct implementation.
Hence, we derive a simplified algorithm from PR splitting method on the dual form to avoid the difficulty.

Here, two intermediate variables $\boldsymbol{p}$ and $\boldsymbol{q}$ are introduced as follows, which is similar to the proof in \cite{zhang2022HPR_OT}, 
\begin{align*}
    \boldsymbol{q}^{k+1} &= \mathbf{J}_{\sigma \partial f}(\boldsymbol{\eta}^k), \\
    \boldsymbol{p}^{k+1} &= \mathbf{J}_{\sigma^{-1}\partial g}\left(\sigma^{-1}(2\sigma \boldsymbol{q}^{k+1} - \boldsymbol{\eta}^k)\right).
\end{align*}
Inspired by \cite{zhang2022HPR_OT}, we develop Algorithm \ref{HPR2_alg} and Theorem \ref{thm_converge}, as an adaptation of the proof described therein. 
Additionally, it can be proved that the intermediate variables $\boldsymbol{p}^k$ and $\boldsymbol{q}^k$ converge to the optimal solution $\boldsymbol{p}^*$ and $\boldsymbol{q}^*$ of the prime problem \eqref{equi_form}, which will be stated in Theorem \ref{thm_converge}. 

\begin{algorithm}[ht]
\caption{PR Splitting for Dual Problem}
\label{HPR2_alg}
\begin{algorithmic}[1]
\STATE \textbf{Input:} $\boldsymbol{x}^0, \boldsymbol{p}^0= \boldsymbol{x}^0$, and $\sigma >0$.
\STATE \textbf{Initialize:} $\boldsymbol{\eta}^0=\boldsymbol{x}^0+\sigma \boldsymbol{p}^0.$
\FOR{$k=0,1,\ldots L\!-\!1$} 
    \STATE $\boldsymbol{q}^{k+1}=(\mathbf{A}^{\top}\mathbf{A}+\sigma \mathbf{I})^{-1}(\mathbf{A}^{\top}\boldsymbol{y}+\boldsymbol{\eta}^{k}).$
    \STATE $\boldsymbol{w}^{k+1}= \boldsymbol{\eta}^k -\sigma\boldsymbol{q}^{k+1}.$
    \STATE $\boldsymbol{p}^{k+1}=\operatorname{Prox}_{{\sigma}^{-1}g}({\sigma}^{-1}(2\sigma \boldsymbol{q}^{k+1}-\boldsymbol{\eta}^k)).$
    \STATE $\boldsymbol{x}^{k+1}=\boldsymbol{\eta}^{k}+\sigma \boldsymbol{p}^{k+1}-2\sigma \boldsymbol{q}^{k+1}.$
    \STATE $\boldsymbol{\eta}^{k+1}=\boldsymbol{\eta}^{k}+2\sigma(\boldsymbol{p}^{k+1}-\boldsymbol{q}^{k+1}).$
\ENDFOR
\STATE \textbf{Output}: $\boldsymbol{x}^L$ 
\end{algorithmic}
\end{algorithm}

\begin{proposition}
\label{conicide_pro}
     If the initialization point \( \boldsymbol{\eta}^0 \) in Algorithm \ref{HPR2_alg} matches \( \boldsymbol{\eta}^0 \) in \eqref{fix_eq}, the sequence $\{ \boldsymbol{\eta}^k\}$ generated by Algorithm \ref{HPR2_alg} coincides with the sequence $\{\boldsymbol{\eta}^k\}$ generated by \eqref{fix_eq}.
\end{proposition}
\begin{proof}
First, as mentioned before, the resolvent be represented by a proximal operator.
Thus, we obtain 
\begin{equation*}
    \begin{aligned}
        \boldsymbol{q}^{k+1}&= \mathbf{J}_{\sigma \partial f}(\boldsymbol{\eta}^k) \\
        &= (\mathbf{A}^{\top}\mathbf{A}+\sigma \mathbf{I})^{-1}(\mathbf{A}^{\top}\boldsymbol{y}+\boldsymbol{\eta}^{k}),\\
\boldsymbol{p}^{k+1}&=\mathbf{J}_{\sigma^{-1}\partial g}(\sigma^{-1}(2\sigma \boldsymbol{q}^{k+1}-\boldsymbol{\eta}^k))\\
&= \operatorname{Prox}_{{\sigma}^{-1}g}({\sigma}^{-1}(2\sigma \boldsymbol{q}^{k+1}-\boldsymbol{\eta}^k)).
    \end{aligned}
\end{equation*}

The statement is then proved by induction. 
For $k=0$, we use the aforementioned definition of $\boldsymbol{q}^1$, yielding the condition $0 \in \partial f(\boldsymbol{q}^1)-\boldsymbol{\eta}^0 + \sigma \boldsymbol{q}^1$. 
Thus, $\boldsymbol{\eta}^0  -\sigma \boldsymbol{q}^1 \in \partial f(\boldsymbol{q}^1)$. 
Invoking Theorem 23.5 in \cite{rockafellar1970convex}, we deduce that
\begin{equation*}
    \boldsymbol{q}^1 \in \partial f^*(\boldsymbol{\eta}^0 - \sigma \boldsymbol{q}^1).
\end{equation*}
Hence, $\sigma \boldsymbol{q}^1 = \sigma \partial f^*(\boldsymbol{\eta}^0 - \sigma \boldsymbol{q}^1)$.
This means that
\begin{equation*}
    \boldsymbol{\eta}^0 \in \boldsymbol{\eta}^0 -\sigma \boldsymbol{q}^1 +\sigma \partial f^*(\boldsymbol{\eta}^0 - \sigma \boldsymbol{q}^1).
\end{equation*}
Since $\mathbf{M}_2 = \partial f^*$, it follows that
\begin{equation*}
    \boldsymbol{\eta}^0 \in \boldsymbol{\eta}^0 - \sigma \boldsymbol{q}^1 + \sigma \mathbf{M}_2(\boldsymbol{\eta}^0 -\sigma  \boldsymbol{q}^1),
\end{equation*}
Let $\boldsymbol{w}^k\triangleq \mathbf{J}_{\sigma \mathbf{M}_2}(\boldsymbol{\eta}^{k-1})$, which implies
\begin{equation*}
    \boldsymbol{w}^1 = J_{\sigma \mathbf{M}_2}(\boldsymbol{\eta}^0) = \boldsymbol{\eta}^0 -\sigma \boldsymbol{q}^1. 
\end{equation*}
In parallel, by the stipulation of $\boldsymbol{p}^1$, we observe that
\begin{equation*}
    0 \in \partial g(\boldsymbol{p}^1) + (\boldsymbol{\eta}^0 - 2\sigma \boldsymbol{q}^1 + \sigma \boldsymbol{p}^1 ).
\end{equation*}
It follows from Theorem 23.5 in \cite{rockafellar1970convex} that
\begin{equation*}
    \boldsymbol{p}^1 \in \partial g^*(-(\boldsymbol{\eta}^0 -2\sigma \boldsymbol{q}^1 + \sigma \boldsymbol{p}^1 )).
\end{equation*}
Hence, $-\sigma \boldsymbol{p}^1 = -\sigma \partial g^*(-(\boldsymbol{\eta}^0 - 2\sigma \boldsymbol{q}^1 + \sigma \boldsymbol{p}^1)$, implying that
{\footnotesize
\begin{equation*}
    2(\boldsymbol{\eta}^0-\sigma \boldsymbol{q}^1) \!-\! \boldsymbol{\eta}^0 \in 2(\boldsymbol{\eta}^0 \!-\! \sigma \boldsymbol{q}^1) \!-\! \boldsymbol{\eta}^0 \!+\! \sigma \boldsymbol{p}^1 \!-\! \sigma \partial g^*(-(\boldsymbol{\eta}^0 \!-\! 2\sigma \boldsymbol{q}^1 \!+\! \sigma \boldsymbol{p}^1)).
\end{equation*}}
Since $\mathbf{M}_1 = \partial (g^*\circ(-\mathbf{I}))$, we have
{\footnotesize
\begin{equation*}
    2(\boldsymbol{\eta}^0 - \sigma \boldsymbol{q}^1) - \boldsymbol{\eta}^0 \in \boldsymbol{\eta}^0 -2\sigma \boldsymbol{q}^1 + \sigma \boldsymbol{p}^1  +\sigma M_1(\boldsymbol{\eta}^0 -2\sigma \boldsymbol{q}^1 + \sigma \boldsymbol{p}^1 ), 
\end{equation*}}
which implies $\boldsymbol{x}^1 := \mathbf{J}_{\sigma \mathbf{M}_1}(2\boldsymbol{w}^1 - \boldsymbol{\eta}^0) = \boldsymbol{\eta}^0 - 2\sigma \boldsymbol{q}^1 + \sigma \boldsymbol{p}^1$. 
Hence, we have $\boldsymbol{\eta}^1 := \boldsymbol{\eta}^0 + 2(\boldsymbol{x}^1 - \boldsymbol{w}^1) = \boldsymbol{\eta}^0 + 2\sigma(\boldsymbol{p}^1 - \boldsymbol{q}^1)$. 
    {\footnotesize
\begin{equation*}
    \begin{aligned}
        \boldsymbol{\eta}^1& = \boldsymbol{\eta}^0 + 2(\boldsymbol{x}^1 - \boldsymbol{w}^1)=\boldsymbol{\eta}^0+2(\mathbf{J}_{\sigma \mathbf{M}_1}(2\boldsymbol{w}^1-\boldsymbol{\eta}^0)-\mathbf{J}_{\sigma \mathbf{M}_2}(\boldsymbol{\eta}^0)) \\
        &=\boldsymbol{\eta}^0+2(\mathbf{J}_{\sigma \mathbf{M}_1}(2\mathbf{J}_{\sigma \mathbf{M}_2}(\boldsymbol{\eta}^0))-\mathbf{J}_{\sigma \mathbf{M}_1}(-\boldsymbol{\eta}^0)-\mathbf{J}_{\sigma \mathbf{M}_2}(\boldsymbol{\eta}^0)) \\
        &=(2\mathbf{J}_{\sigma \mathbf{M}_1}-\mathbf{I})(2\mathbf{J}_{\sigma \mathbf{M}_2}-\mathbf{I})(\boldsymbol{\eta}^0)=\mathbf{R}_{\sigma \mathbf{M}_1}\mathbf{R}_{\sigma \mathbf{M}_2}\boldsymbol{\eta}^0.
    \end{aligned}
\end{equation*}
}

It follows that the update of $\boldsymbol{\eta}^1, \boldsymbol{w}^1$ in Algorithm \ref{HPR2_alg} are the same as that in \eqref{fix_eq}. 
Hence, we prove the statement for $k = 0$. 
Assume that the statement holds for some $k \ge 1$. 
For $k := k + 1$, we can prove that the statement holds similarly to the case $k = 0$. 
Thus, we prove the statement holds for any $k \ge 0$ by induction.
\end{proof}

In the following, we demonstrate the convergence of the sequences $\{\boldsymbol{x}^k\}$ and $\{\boldsymbol{p}^k\}$ in Algorithm \ref{HPR2_alg} to the respective optimal points of the dual and primal problems, thereby confirming the algorithm's effectiveness.

\begin{theorem}
\label{thm_converge}
    The sequence $\{\boldsymbol{w}^k\}$ and $\{\boldsymbol{x}^k\}$ generated by Algorithm \ref{HPR2_alg} converge to the optimal solution $\boldsymbol{w}^*$ of \eqref{dual_pro}. 
    Moreover, when the regularization term $g(\boldsymbol{p})$ in \eqref{equi_form} is assumed convex, the sequence $\{\boldsymbol{p}^k\}$, $\{\boldsymbol{q}^k\}$generated by Algorithm \ref{HPR2_alg} converge to the optimal solution $\boldsymbol{p}^*$, $\boldsymbol{q}^*$ of \eqref{equi_form}.
\end{theorem}
\begin{proof}
Since $\operatorname{Fix}\left(\mathbf{R}_{\sigma \mathbf{M}_1}\mathbf{R}_{\sigma \mathbf{M}_2}\right)$ is a closed convex set \cite[Corollary 4.24]{Bauschke2011monotone_book}, the fixed-point is unique. 
Using the convergence result of \cite[Theorem 22]{davydov2022non_monotone}, we can directly obtain that the sequence $\{\boldsymbol{\eta}^k\}$ converges to the fixed-point $\boldsymbol{\eta}^*$ of the PR iteration. 
Since the resolvent $\mathbf{J}_{\sigma \mathbf{M}_2}$ is non-expansive
we directly obtain 
$$
\lim_k \boldsymbol{w}^k=\lim_k \mathbf{J}_{\sigma \mathbf{M}_2}\left(\boldsymbol{\eta}^k\right)=\mathbf{J}_{\sigma \mathbf{M}_2}\left(\boldsymbol{\eta}^*\right).
$$
As show in \cite{lions1979splitting}, if $\boldsymbol{\eta}^*$ is the fixed-point of the iteration, then $\boldsymbol{w}^*=\mathbf{J}_{\sigma \mathbf{M}_2}\left(\boldsymbol{\eta}^*\right)$.
Additionally, it can be obtained that
$$\boldsymbol{x}^{k+1}-\boldsymbol{w}^{k+1}=\frac{1}{2}(\boldsymbol{\eta}^{k+1}-\boldsymbol{\eta}^k)\rightarrow 0.$$
Thus, $\boldsymbol{x}^k$ also converges to to the optimal solution $\boldsymbol{w}^*$ of \eqref{dual_pro}.

Next, we prove the convergence of the sequence $\left\{\boldsymbol{q}^k\right\}$. 
According to Algorithm \ref{HPR2_alg}, we have for all $k \geq 0$, 
$$
\boldsymbol{q}^{k+1}=(\mathbf{A}^{\top}\mathbf{A}+\sigma \mathbf{I})^{-1}(\mathbf{A}^{\top} \boldsymbol{y}+\boldsymbol{\eta}^{k}).
$$
Denote $\hat{f}(\boldsymbol{q}):=f(\boldsymbol{q})+\frac{\sigma}{2}\left\|\boldsymbol{q}\right\|^2$, which is a strongly convex function. 
Thus, $\hat{f}^*$ is essentially smooth \cite[Theorem 26.3]{rockafellar1970convex}. 
The first-order optimality condition of the above iteration of $\boldsymbol{q}^{k+1}$ implies $0 \in \partial \hat{f}\left(\boldsymbol{q}^{k+1}\right)-\boldsymbol{\eta}^k$. 
Since $\hat{f}$ is a proper closed convex function, by \cite[Theorem 23.5]{rockafellar1970convex}, the first-order optimality condition is equivalent to
$$
\boldsymbol{q}^{k+1}=\nabla \hat{f}^*\left(\boldsymbol{\eta}^k\right).
$$
It follows from the convergence of $\left\{\boldsymbol{\eta}^k\right\}$ and the continuity of $\nabla \hat{f}^*$ \cite[Theorem 25.5]{rockafellar1970convex} that $\left\{\boldsymbol{q}^k\right\}$ is convergent. 
Note that 
$$
\boldsymbol{\eta}^{k+1}=\boldsymbol{\eta}^{k}+2\sigma(\boldsymbol{p}^{k+1}-\boldsymbol{q}^{k+1})
$$ 
from Algorithm \ref{HPR2_alg}, and the sequence $\{\boldsymbol{\eta}^k\}$ is convergent, we obtain $\lim_k (\boldsymbol{p}^k-\boldsymbol{q}^k)=0$. 
Then, by taking the limit, we obtain $\lim_k \boldsymbol{p}^k = \lim_k \boldsymbol{q}^k= \boldsymbol{q}^*$. 
Hence, $\{\boldsymbol{p}_k\}$ is convergent.

Assume that $\left(\boldsymbol{p}^*, \boldsymbol{q}^*, \boldsymbol{x}^*\right)$ is the limit point of the sequence $\left\{\boldsymbol{p}^k, \boldsymbol{q}^k, \boldsymbol{x}^k\right\}$. 
Since $\boldsymbol{\eta}^{k+1}=\boldsymbol{\eta}^{k}+2\sigma(\boldsymbol{p}^{k+1}-\boldsymbol{q}^{k+1})$ from Algorithm \ref{HPR2_alg}, we can obtain $\boldsymbol{p}^*-\boldsymbol{q}^*=\boldsymbol{0}$ by taking the limit. 
It follows that 
$$
\lim_k \boldsymbol{w}^{k}=\lim_k [\boldsymbol{w}^k+\sigma(\boldsymbol{p}^k-\boldsymbol{q}^k)]=\lim_k \boldsymbol{x}^k=\boldsymbol{x}^*.
$$
By Algorithm \ref{HPR2_alg}, we have
\begin{equation*}
    0 \in \partial f(\boldsymbol{q}^{k+1})-\boldsymbol{\eta}^{k}+\sigma \boldsymbol{q}^{k+1}.
\end{equation*}
Since $\boldsymbol{w}^{k+1}=\boldsymbol{\eta}^{k} -\sigma \boldsymbol{q}^{k+1}$, we have
\begin{equation*}
    \boldsymbol{w}^{k+1} \in \partial f\left(\boldsymbol{q}^{k+1}\right).
\end{equation*}
Similarly, we have 
\begin{equation*}
    0 \in \partial g(\boldsymbol{p}^{k+1}) + (\boldsymbol{\eta}^k - 2\sigma \boldsymbol{q}^{k+1} + \sigma \boldsymbol{p}^{k+1} ).
\end{equation*}
Since $\boldsymbol{x}^{k+1} = \boldsymbol{\eta}^k - 2\sigma \boldsymbol{q}^{k+1} + \sigma \boldsymbol{p}^{k+1}$, we have 
\begin{equation*}
    -\boldsymbol{x}^{k+1} \in \partial g(\boldsymbol{p}^{k+1}).
\end{equation*}
Hence, we have 
$$
\boldsymbol{w}^{k+1} \in \partial f\left(\boldsymbol{q}^{k+1}\right), \quad-\boldsymbol{x}^{k+1} \in \partial g\left(\boldsymbol{p}^{k+1}\right).
$$
Together with $\boldsymbol{p}^*-\boldsymbol{q}^*=\boldsymbol{0}$ and taking limit, we have
$$
\boldsymbol{x}^* \in \partial f\left(\boldsymbol{q}^*\right), \quad -\boldsymbol{x}^* \in \partial g\left(\boldsymbol{p}^*\right), \quad \boldsymbol{p}^*-\boldsymbol{q}^*=0.
$$
This completes the proof \cite[Corollary 28.3.1]{rockafellar1970convex}.  
\end{proof}


\section{Implementing PR Splitting Method with Deep Equilibrium Network}
\label{sec_Implementing PR Splitting Method with Deep Equilibrium Network}

In the previous discussion, a fixed-point equation with non-expansive property is constructed to solve the dual problem, ensuring controllable Lipschitz constants suitable for the DEQ model.  
Consequently, we implement this using a single fixed layer to adapt the fixed-point equation. 
Specifically, we unfold one single iteration of Algorithm \ref{HPR2_alg} by replacing the proximal operator with a trainable NN, thereby modeling the fixed-point equation with a fixed layer.
This approach not only inherits DEQ's low complexity but also ensures theoretical convergence and optimality of the PR splitting approach.

\subsection{NN Approximation of the Nonlinear Proximal Operator}

Given the neural network's strong ability to approximate nonlinear functions and capture statistical features, it is practical to use NNs to learn the nonlinear proximal operator, which is correlated with the prior distribution of the channel. 

Similar to the deep unfolding approach of \cite{zhang2018ista_net}, we consider replacing $\operatorname{Prox}_{\sigma^{-1} g}$ with a trainable NN $R_\theta$. 
Specifically, $R_{\theta}$ is constructed using a classical residual network \cite{he2016deep}, which facilitates efficient training of deeper models by leveraging shortcut connections for learning residual mappings. The residual network is mainly composed of four residual blocks, each of which is activated by two convolution layers and two RELU functions. More specifically, the convolutional layers are equipped with $3\!\times\! 3$ kernels and a fixed number of $64$ feature maps. 

\begin{figure}[b]
\centering
\includegraphics[width=1\linewidth]{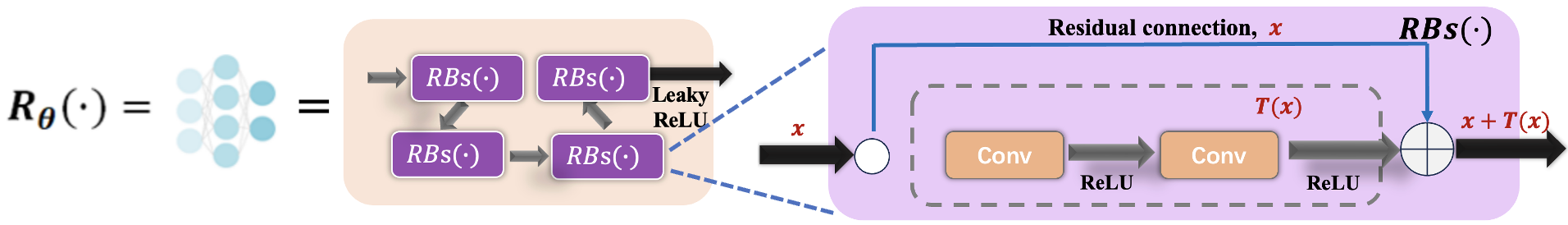}
\caption{The design of a NN block $R_\theta$ to approximate the proximal operator, denoted by $R_\theta \rightarrow \operatorname{Prox}_{\sigma^{-1}g}$.}
\label{Rrox}
\end{figure}

\subsection{Deep Equilibrium Network}

The Deep Equilibrium (DEQ) model is a neural architecture that trains a single fixed layer to model the equilibrium state of a fixed-point equation, enabling efficient inference \cite{bai2019DEM}. 
We adopt the DEQ model to learn the network parameter $\theta$ since, compared to traditional deep unfolding methods, it not only ensures theoretical guarantees from its design of a fixed-point equation but also significantly reduces memory overhead and training complexity \cite{yu2023adaptive,zhao2022deep_snapshot,bai2019DEM}.

The iterative variable $\boldsymbol{\eta}^k$ in DEQ model can be represented as:
$$
\begin{aligned}
\lim _{k \rightarrow+\infty} \boldsymbol{\eta}^{(k)} & =\lim _{k \rightarrow+\infty} f_\theta\left(\boldsymbol{\eta}^{(k)} ; \boldsymbol{y}\right) \\
& \equiv \hat{\boldsymbol{\eta}}=f_\theta(\hat{\boldsymbol{\eta}} ; \boldsymbol{y}),
\end{aligned}
$$
where $\hat{\boldsymbol{\eta}}$ denotes the fixed-point in the network and $f_\theta$ denotes the fixed layer of the neural architecture with a fixed parameter $\theta$, which contains the the trainable network $R_{\theta}$.

To optimize network parameters $\theta$, stochastic gradient descent \cite{yu2022hybrid} is used to minimize a loss function as follows:
$$
\theta^*=\arg \min _\theta \frac{1}{m} \sum_{i=1}^m \ell\left(f_\theta\left(\hat{\boldsymbol{x}}_i ; \boldsymbol{y}_i, \right), \boldsymbol{h}_i^{*}\right),
$$
where $m$ is the number of training samples,  
$\hat{\boldsymbol{x}}_i$ denotes the fixed-point generated by the network iteration, $\boldsymbol{h}_i^{*}$ is the ground truth channel of the $i$-th training sample, and $\boldsymbol{y}_i$ is the paired measurement. 
$\ell(\cdot, \cdot)$ is the loss function, defined by the mean squared error (MSE), as
$$
\ell\left(\hat{\boldsymbol{x}}, \boldsymbol{h}^{*}\right)=\frac{1}{2}\left\|\hat{\boldsymbol{x}}-\boldsymbol{h}^{*}\right\|_2^2.
$$

Then, we calculate the loss gradient. 
Let $\ell$ be an abbreviation of $\ell\left(\hat{\boldsymbol{x}}, \boldsymbol{h}^{\ast}\right)$, then the loss gradient is \cite{bai2019DEM}:
{
\small
\begin{equation*}
    \frac{\partial \ell}{\partial \theta}=\left[\frac{\partial f_\theta(\hat{\boldsymbol{x}} ; \boldsymbol{y})}{\partial \theta}\right]^{\top}\left[\boldsymbol{I}-\left.\frac{\partial f_\theta(\boldsymbol{x} ; \boldsymbol{y})}{\partial \boldsymbol{x}}\right|_{\boldsymbol{x}=\hat{\boldsymbol{x}}}\right]^{-\top}\left(\hat{\boldsymbol{x}}-\boldsymbol{h}^{\ast}\right),
\end{equation*}
}
where ${ }^{-\top}$ denotes the inversion followed by transpose. 

Using DEQ model, the memory complexity of our approach stays at $O(1)$ \cite{zhao2022deep_snapshot}, independent of iteration count, and notably lower than deep unfolding's $O(T)$, enabling efficient gradient calculation for the loss term with minimal memory demand. 

\subsection{Implementation and Convergence Analysis}
In the following, we implement PR splitting method with the DEQ model by a fixed layer $f_\theta$, constructing a deep learning-based estimator for channel estimation. 
%
Unfolding a single iteration of Algorithm \ref{HPR2_alg}, it can be observed that the algorithm consists of three effective linear terms (steps 4, 7, and 8) and one nonlinear term (step 6) within each loop. 
Thus, by replacing the the nonlinear proximal term with the aforementioned NN block $R_\theta$, which is denoted as $R_\theta \rightarrow \operatorname{Prox}_{\sigma^{-1}g}$, the iterative process of variable $\boldsymbol{\eta}^k$ in Algorithm \ref{HPR2_alg} is transformed into the following expression:
\begin{equation}
  \label{iter_MEDOS}
  \boldsymbol{\eta}^{k+1}_R = f_\theta(\boldsymbol{\eta}^k_R)=f_{\text{LT}_3}\circ f_{\text{LT}_2}\circ R_\theta \circ f_{\text{LT}_1}(\boldsymbol{\eta}^k_R).
\end{equation}
Here, $f_{\text{LT}_i}$denotes the $i$-th linear term (LT) in Algorithm \ref{HPR2_alg}, and $\boldsymbol{\eta}^k_R$ denotes the intermediate variable produced by the network. 
Fig. \ref{DR-PR-net} illustrates the detailed iterative process, where $\boldsymbol{\eta}^k$ iterates to a fixed point through the fixed layer, resulting in $\boldsymbol{x}^L$ as the final channel estimation.
The estimator with the whole network architecture is called as Peaceman-Rachford splitting method implemented with Deep Equilibrium Network (denoted as PR-DEN for short).

%
\begin{figure}[b]
\centering
\includegraphics[width=0.85\linewidth]{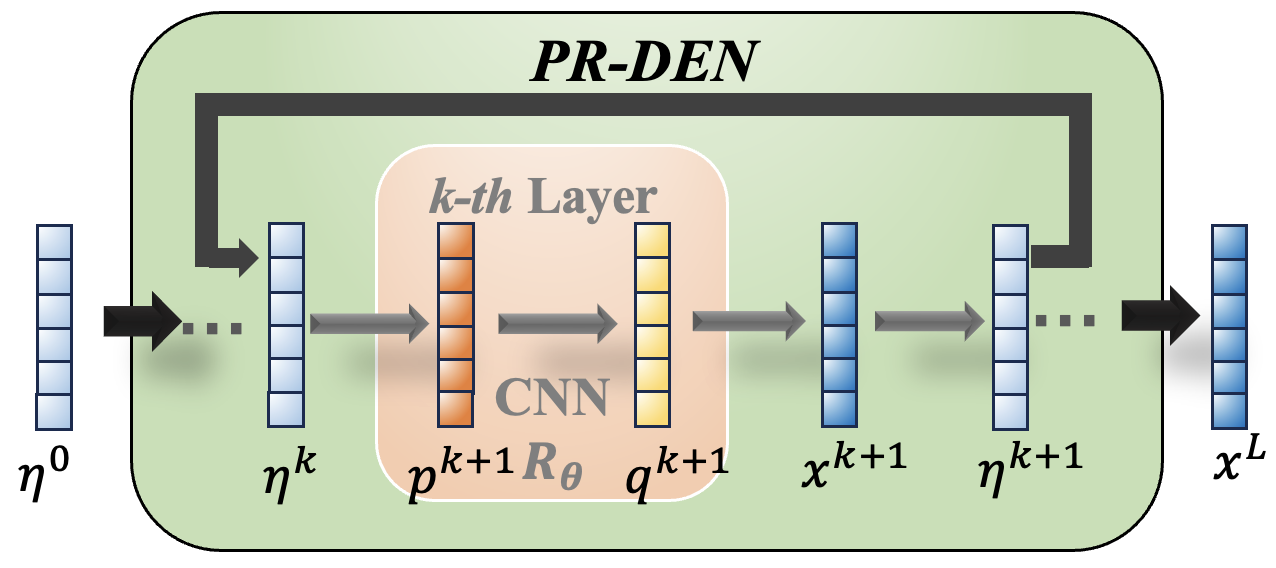}
\caption{Peaceman-Rachford splitting method implemented with Deep Equilibrium Network (PR-DEN) approach.}
\label{DR-PR-net}
\end{figure}

Consequently, we focus on demonstrating the theoretical advantages of PR-DEN, highlighting its strict convergence guarantee and the optimality of the converged solution.
Referring to Remark \ref{remark_nonexpansive}, it can be deduced that $\text{Lip}(\mathbf{R}_{\sigma \mathbf{M}_1}\mathbf{R}_{\sigma \mathbf{M}_2})\leq 1$. Without loss of generality, assume $\text{Lip}(\mathbf{R}_{\sigma \mathbf{M}_1}\mathbf{R}_{\sigma \mathbf{M}_2})<1$, since when $\text{Lip}(\mathbf{R}_{\sigma \mathbf{M}_1}\mathbf{R}_{\sigma \mathbf{M}_2})=1$, a perturbation parameter $\varepsilon$ can be added to construct a contraction mapping \cite[Proposition 4.20]{Bauschke2011monotone_book} and the sequence of iteration converge to the same solution, which leads to the following Theorem.

\begin{theorem} \label{thm_lpi}
There exists a NN $R_{\theta}(\cdot)$ which approximates the corresponding proximal operator, such that the sequence $\{\boldsymbol{\eta}^k_R\}$ produced by PR-DEN is strictly convergent with a linear convergence rate.

\end{theorem}
\begin{proof}
As indicated above, the iteration formula for $\boldsymbol{\eta}^k$ in the proposed PR-DEN can be succinctly represented as: 
\begin{equation*}
\label{iter_MEDOS2}
    \boldsymbol{\eta}^{k+1}_R =f_{\text{LT}_3}\circ f_{\text{LT}_2}\circ R_\theta \circ f_{\text{LT}_1}(\boldsymbol{\eta}^k_R).
\end{equation*}
That is, 
\begin{equation*}
\begin{aligned}
    \boldsymbol{\eta}^{k+1}_R =&f_{\text{LT}_3}\!\circ\! f_{\text{LT}_2}\!\circ\!(R_\theta-\operatorname{Prox}_{\sigma^{-1}g}+\operatorname{Prox}_{\sigma^{-1}g}) \!\circ\! f_{\text{LT}_1}(\boldsymbol{\eta}^k_R)\\
    =&f_{\text{LT}_3}\circ f_{\text{LT}_2}\circ( R_\theta-\operatorname{Prox}_{\sigma^{-1}g })\circ f_{\text{LT}_1}(\boldsymbol{\eta}^k_R)\\
    &+f_{\text{LT}_3}\circ f_{\text{LT}_2}\circ \operatorname{Prox}_{\sigma^{-1}g }\circ f_{\text{LT}_1}(\boldsymbol{\eta}^k_R).
\end{aligned}
\end{equation*}
From the proof of proposition \ref{conicide_pro}, we obtain 
\begin{equation*}
    \mathbf{R}_{\sigma \mathbf{M}_1}\mathbf{R}_{\sigma \mathbf{M}_2}(\boldsymbol{\eta}^k_R)=f_{\text{LT}_3}\circ f_{\text{LT}_2}\circ \operatorname{Prox}_{\sigma^{-1}g }\circ f_{\text{LT}_1}(\boldsymbol{\eta}^k_R).
\end{equation*}
Thus, we derive
\begin{equation}
\label{eta^k_R}
\begin{aligned}
        \boldsymbol{\eta}^{k+1}_R =&f_{\text{LT}_3}\circ f_{\text{LT}_2}\circ( R_\theta-\operatorname{Prox}_{\sigma^{-1}g })\circ f_{\text{LT}_1}(\boldsymbol{\eta}^k_R)\\
        +&\mathbf{\mathbf{R}_{\sigma \mathbf{M}_1}\mathbf{R}_{\sigma \mathbf{M}_2}}(\boldsymbol{\eta}^k_R).
\end{aligned}
\end{equation}

By the Universal approximation theorem \cite{kidger2020universal}, there exists a network $R_{\theta}(\cdot)$ approximating the proximal operator $\operatorname{Prox}_{\sigma^{-1}g}$ with arbitrary precision. 
Suppose $\max\{\|f_{\text{LT}_1}\|,\|f_{\text{LT}_2}\|,\|f_{\text{LT}_3}\|\}=\alpha$, $\text{Lip}(\mathbf{\mathbf{R}_{\sigma \mathbf{M}_1}\mathbf{R}_{\sigma \mathbf{M}_2}})=\beta$. There exists $R_\theta$,  such that
\begin{equation}
\label{appro}
   \|(R_\theta-\operatorname{Prox}_{\sigma^{-1}g }) (\boldsymbol{\eta}^k_R)\|<\frac{1-\beta}{\alpha^3}\|\boldsymbol{\eta}^k_R\|.
\end{equation}
Thus, combined \eqref{appro} with \eqref{eta^k_R}, we derive
\begin{equation*}
    \|\boldsymbol{\eta}^{k+1}_R\|<(1-\beta) \|\boldsymbol{\eta}^k_R\|+\beta \|\boldsymbol{\eta}^k_R\|=\|\boldsymbol{\eta}^k_R\|.
\end{equation*}
From Banach fix point theorem, $\{\boldsymbol{\eta}^k_R\}$ is convergent.
\end{proof}

Up to now, PR-DEN approach is proved to be convergent with a linear convergence rate. Consequently, it can be proved that the iterative solution produced by PR-DEN is convergent to the optimal solution of dual problem \eqref{dual_pro}.

\begin{theorem}
\label{thm_prden_opt}
    There exists a NN $R_\theta$ that sufficiently approximates $\operatorname{Prox}_{{\sigma}^{-1}g}$, such that $\boldsymbol{x}_R$ converges to $\boldsymbol{w}^*$, where $\boldsymbol{x}_R$ is the solution produced by PR-DEN, and $\boldsymbol{w}^*$ is the optimal solution of dual problem \eqref{dual_pro}. More specifically, when the regularization term $g(\boldsymbol{p})$ is assumed to be convex, $\boldsymbol{x}_R$ converges to the optimal solution $\boldsymbol{p}^*$ of \eqref{equi_form}.
\end{theorem}

\begin{proof}
From Theorem 2, we derive the convergence of $\{\boldsymbol{\eta}^k_{R}\}$, which denotes the sequence produced by PR-DEN. That is 
\begin{equation*}
    \lim_{k \rightarrow \infty} \boldsymbol{\eta}^k_R=\boldsymbol{\eta}_R,
\end{equation*}
where $\boldsymbol{\eta}_R$ is the solution of the following fix point equation:
\begin{equation}
\label{eta_R}
    \boldsymbol{\eta}_R=f_{\text{LT}_3}\circ f_{\text{LT}_2}\circ R_\theta \circ f_{\text{LT}_1}(\boldsymbol{\eta}_R).
\end{equation}

Additionally, the optimal solution $\boldsymbol{w}^*$ is equivalent to the solution of the following problem \cite{lions1979splitting}:
\begin{equation}
\label{eta_*}
    \boldsymbol{w}^*=\mathbf{J}_{\sigma \mathbf{M_2}}(\boldsymbol{\eta}^*) \quad\text{s.t.} \quad
    \boldsymbol{\eta}^*=\mathbf{R}_{\sigma \mathbf{M}_1}\mathbf{R}_{\sigma \mathbf{M}_2}(\boldsymbol{\eta}^*).
\end{equation}
Thus, subtracting $\boldsymbol{\eta}_R$ in \eqref{eta_R} from $\boldsymbol{\eta}^*$ in \eqref{eta_*}, we obtain
\begin{equation*}
\begin{aligned}
        \boldsymbol{\eta}_R-\boldsymbol{\eta}^*&=\mathbf{\mathbf{R}_{\sigma \mathbf{M}_1}\mathbf{R}_{\sigma \mathbf{M}_2}}(\boldsymbol{\eta}_R-\boldsymbol{\eta}^*)\\
        &+f_{\text{LT}_3}\circ f_{\text{LT}_2}\circ( R_\theta-\operatorname{Prox}_{\sigma^{-1}g })\circ f_{\text{LT}_1}(\boldsymbol{\eta}_R).
\end{aligned}
\end{equation*}

Suppose $$\max\{\|f_{\text{LT}_1}\|,\|f_{\text{LT}_2}\|,\|f_{\text{LT}_3}\|\}=\alpha,$$ and $\text{Lip}(\mathbf{\mathbf{R}_{\sigma \mathbf{M}_1}\mathbf{R}_{\sigma \mathbf{M}_2}})=\beta$, then
\begin{equation*}
    \| \boldsymbol{\eta}_R-\boldsymbol{\eta}^*\|\leq \beta\| \boldsymbol{\eta}_R-\boldsymbol{\eta}^*\|+\alpha^3\|(R_\theta-\operatorname{Prox}_{\sigma^{-1}g })\|\|\boldsymbol{\eta}_R\|\\.
\end{equation*}
Hence, we have
\begin{equation*}
    \| \boldsymbol{\eta}_R-\boldsymbol{\eta}^*\|\leq \frac{\alpha^3}{1-\beta}\|(R_\theta-\operatorname{Prox}_{\sigma^{-1}g})\|\|\boldsymbol{\eta}_R\|.
\end{equation*}
We claim there exists a consistent upper bound $M$ on $\boldsymbol{\eta}_R$. 
\begin{equation}
\label{upper}
    \|\boldsymbol{\eta}_R\|\leq\|f_{\text{LT}_3}\circ f_{\text{LT}_2}\circ R_\theta \circ f_{\text{LT}_1}(\boldsymbol{\eta}_R)\|+1.
\end{equation}

By the Universal approximation theorem \cite{kidger2020universal}, there exists a network $R_{\theta}(\cdot)$ approximating the proximal operator $\operatorname{Prox}_{\sigma^{-1}g}$ with arbitrary precision. This indicates that for any $\boldsymbol{\eta}$ and $\varepsilon>0$ we have 
\begin{equation*}
    \|f_{\text{LT}_3}\circ f_{\text{LT}_2}\circ (R_\theta-\operatorname{Prox}_{\sigma^{-1}g}) \circ f_{\text{LT}_1}(\boldsymbol{\eta})\|<\epsilon.
\end{equation*}
when $R_\theta$ sufficiently approximates $\operatorname{Prox}_{\sigma^{-1}g}$.

It demonstrates that 
$$
\text{Lip}(f_{\text{LT}_3}\circ f_{\text{LT}_2}\circ R_\theta \circ f_{\text{LT}_1})<1.
$$
Combined with \eqref{upper}, it convinces the fact that $\|\boldsymbol{\eta}_R\|$ is bounded. 
Here, we denote the upper bound as $M$. 
Thus, 
\begin{equation*}
    \| \boldsymbol{\eta}_R-\boldsymbol{\eta}^*\|\leq \frac{\alpha^3M}{1-\beta}\|(R_\theta-\operatorname{Prox}_{\sigma^{-1}g})\|.
\end{equation*}
Since $R_\theta\rightarrow\operatorname{Prox}_{\sigma^{-1}g}$, it can be indicated that 
\begin{equation*}
    \boldsymbol{\eta}_R\rightarrow \boldsymbol{\eta}^* .
\end{equation*}
Using a similar approach, we can derive
\begin{equation*}
  \boldsymbol{x}_R\rightarrow\boldsymbol{w}^*.
\end{equation*}
Additionally, the convexity ensures that $\boldsymbol{p}^*=\boldsymbol{w}^*$, which completes the proof.
\end{proof}

\section{Numerical Results} 

This section presents the simulation results of PR-DEN as shown in Section~\ref{sec_Implementing PR Splitting Method with Deep Equilibrium Network}. 
Specifically, the simulations are conducted using the hybrid far- and near-field channel model illustrated in subsection~\ref{subsec_Hybrid-Field Channel Model and Parameter Settings}.
The detailed parameter settings are also shown in subsection~\ref{subsec_Hybrid-Field Channel Model and Parameter Settings}. We present the simulation results in subsection~\ref{subsec_Simulation Results}.
The source code is publicly available on GitHub \footnote{https://github.com/wushitong1234/PR-DEN}. 
%

\subsection{Hybrid-Field Channel Model}
\label{subsec_Hybrid-Field Channel Model and Parameter Settings}
In the following simulation, we employ a hybrid far- and near-field channel model for an ultra-massive MIMO system. This scenario presents a high-dimensional channel that is a stochastic mixture of far- and near-field components, making channel estimation particularly challenging \cite{yu2022hybrid}. By considering such a complex setting, we aim to demonstrate the performance advantages of our approach.

The generation of far-field or near-field channel is determined by the Rayleigh distance, \textit{i.e.} $D_\text{Rayleigh}=\frac{D^2}{\lambda_c}$, where $D$ is the array aperture and and $\lambda_c$ is the carrier wavelength. 
The propagation channel is represented by a ray-based model of $L$ rays, where the $0$-th ray corresponds to the line-of-sight (LoS) path, while the $l$-th ($l=1,\cdots, L-1$) ray is non-line-of-sight (NLoS) path. 
Specifically, the channel is expressed as \cite{dovelos2021channel}:
\begin{equation*}
    \boldsymbol{h}= \textstyle{\sum_{l=0}^{L}} \beta_{l} \mathbf{a}\left(\phi_{l}, \theta_{l}, r_{l}\right) e^{-j 2 \pi f_{c} \tau_{l}}
\end{equation*}
where the parameter $f_{c}$ is the carrier frequency. 
And the parameters $\beta_{l}, \phi_{l}, \theta_{l}, r_{l}, \mathbf{a}\left(\phi_{l}, \theta_{l}, r_{l}\right)$, and $\tau_{l} $ are respectively the path loss, azimuth angle of arrival (AoA), elevation AoA, distance between the array and the RF source/scatterer, array response vector, and time delay of the $l$-th path. 
The path loss can be represented by the following expression \cite{dovelos2021channel}:
$$
\beta_{l}=\left|\Gamma_{l}\right|\left(\frac{c}{4 \pi f_{c} r_{1}}\right) e^{-\frac{1}{2} k_{\mathrm{abs}} r_{1}},
$$
where $\Gamma_{l}$ is the reflection coefficient,  $ r_{1} $ is the LoS path length, and $k_{\mathrm{abs}}$ is the molecular absorption coefficient and 
\begin{equation*}
    \Gamma_{l} = \begin{cases}
    \scriptstyle\frac{\cos \varphi_{\mathrm{in}, l} - n_{t} \cos \varphi_{\mathrm{ref}, l}}{\cos \varphi_{\mathrm{in}, l} + n_{t} \cos \varphi_{\mathrm{ref}, l}} e^{-\left(\frac{8 \pi^{2} f_{c}^{2} \sigma_{\mathrm{rough}}^{2} \cos ^{2} \varphi_{\mathrm{in}, l}}{c^{2}}\right)}, & \text{if } l > 1, \\
    1, & \text{if } l = 0.
    \end{cases}
\end{equation*}
Here $\varphi_{\mathrm{in}, l}$ is the angle of incidence of the $l$-th path, $\varphi_{\mathrm{ref}, l} = \arcsin \left(n_{t}^{-1} \sin \varphi_{\mathrm{in}, l}\right)$ is the angle of refraction, $ n_{t}$ is the refractive index and $\sigma_{\text {rough }}$ is the roughness coefficient of the reflecting material. 

Since the wavefront is approximately planar in the far field and spherical in the near field, the array responses, which are determined by the distance $r_l$, can be expressed as\cite{yu2022hybrid}: 
$$
\mathbf{a}\left(\phi_{l}, \theta_{l}, r_{l}\right)=\left\{\begin{array}{ll}
\text{vec}(\mathbf{A}^{\text {far }}\left(\phi_{l}, \theta_{l}\right)), & \text { if } r_{l}>D_{\text {Rayleigh }}, \\
\text{vec}(\mathbf{A}^{\text {near }}\left(\phi_{l}, \theta_{l}, r_{l}\right) )& \text { otherwise. }
\end{array}\right.
$$
where $\text{vec}()$ denotes the operation of unfolding a matrix into a vector. Specifically, 
\begin{equation*}
    \begin{cases}
    \mathbf{A}^{\text {far }}\left(\phi_{l}, \theta_{l}\right)_{s_1,s_2}=e^{-j 2 \pi \frac{f_{c}}{c} \mathbf{p}_{s_1,s_2}^{T} \mathbf{t}_{l}},\\
    \mathbf{A}^{\text {near }}\left(\phi_{l}, \theta_{l}, r_{l}\right)_{s_1,s_2}=e^{-j 2 \pi \frac{f_{c}}{c}\left\|\mathbf{p}_{s_1,s_2}-r_{l} \mathbf{t}_{l}\right\|_{2}} .
    \end{cases}
\end{equation*}
where $c$ is the speed of light, and $\mathbf{t}_{l}$ is the unit-length vector in the AoA direction of the $l$-th path, given by $\mathbf{t}_{l}=\left[\sin \theta_{l} \cos \phi_{l}, \sin \theta_{l} \sin \phi_{l}, \cos \theta_{l}\right]^{T}$. 
The coordinate $\mathbf{p}_{s_1,s_2}$ can be expressed as \cite{yu2022hybrid}: 
\begin{equation*}
\begin{aligned}
 \mathbf{p}_{s_1,s_2}=[(m_1-1)d+(m_2-1)(\sqrt{N/S}-1)d_\text{UPA},\\ (n_1-1)d+(n_2-1)(\sqrt{N/S}-1)d_\text{UPA}, 0]^{T}
\end{aligned}
\end{equation*} 
where
$s_1=(m_2-1)(\sqrt{N/S})+n_2$, $s_2=(m_1-1)(\sqrt{N/S})+n_1$, $1\leq m_2,n_2 \leq \sqrt{S}$, and $1\leq m_1,n_1 \leq \sqrt{N/S}$. $S$ refers to the number of UPAs. $d_{\text{UPA}}$ refers to the separation between adjacent UPAs, and $d$ denotes the antenna spacing.

In the simulation, the BS is equipped with 4 UPAs, each containing 256 antennas, for a total of 1024 antennas. The digital combining matrix is set as an identity matrix. Additional key simulation parameters are provided in Table \ref{table2}.

\begin{table}[t]
\caption{Parameter Setting}
\label{table2}
\centering
\begin{tabular}{|l|l|}
\hline
\textbf{Parameter}                   & \textbf{Value}                    \\ \hline
Number of RF chain/ UPA        & $N_{\text{RF}} / \quad S= 4$               \\
Number of BS antennas                & $N = 1024$                       \\
Number of antennas in each UPA       & $N/S = 256$                         \\
Carrier frequency                    & $f_c = 300 \, \mathrm{GHz}$       \\
Antenna spacing                      & $d = 5.0 \times 10^{-4} \, \mathrm{m}$ \\
UPA spacing                     & $d_{\text{UPA}} = 5.6 \times 10^{-2} \, \mathrm{m}$ \\
Pilot length                         & $P = 128$                         \\
Azimuth AoA                          & $\theta_{l} \sim \mathcal{U}(-\pi / 2, \pi / 2)$ \\
Elevation AoA                        & $\phi_{l} \sim \mathcal{U}(-\pi, \pi)$ \\
Angle of incidence                   & $\varphi_{\text{in}, l} \sim \mathcal{U}(0, \pi / 2)$ \\
Number of paths                      & $L = 5$                           \\
Rayleigh distance                    & $D_{\text{Rayleigh}} = 20 \, \mathrm{m}$ \\
LoS path length                      & $r_{1} = 30 \, \mathrm{m}$        \\
Scatterer distance ($l > 1$)         & $r_{l} \sim \mathcal{U}(10,25) \, \mathrm{m}$ \\
Time delay of LoS path               & $\tau_{1} = 100 \, \mathrm{nsec}$ \\
Time delay of NLoS paths ($l > 1$)   & $\tau_{l} \sim \mathcal{U}(100,110) \, \mathrm{nsec}$ \\
Absorption coefficient               & $k_{\text{abs}} = 0.0033 \, \mathrm{m}^{-1}$ \\
Refractive index                     & $n_{t} = 2.24 - j0.025$           \\
Roughness factor                     & $\sigma_{\text{rough}} = 8.8 \times 10^{-5} \, \mathrm{m}$ \\ \hline
\end{tabular}
\end{table}

\subsection{Simulation Results}
\label{subsec_Simulation Results}
Our performance metric is the Normalized Mean Square Error (NMSE), defined as
$
\text{NMSE}=10\log_{10}\frac{\|\mathbf{h}-\hat{\mathbf{h}}\|}{\|\mathbf{h}\|},
$ 
where $\mathbf{h}$ denotes the testing channel while $\hat{\mathbf{h}}$ is the estimated channel in the simulation. 

We calculate the NMSE of comprehensive testing datasets comprising 5,000 samples.
The benchmarks include (a) the classical channel estimation methods with closed-form expression: \textbf{LS}, \textbf{MMSE} \cite{cho2010mimo}, (b) classical iterative algorithm including \textbf{OAMP} \cite{ma2017orthogonal}, \textbf{FISTA} \cite{beck2009fista}, and (c) two NN-based iterative algorithms: \textbf{ISTA-Net+} \cite{zhang2018ista_net}, a classical deep unfolding method and \textbf{FPN-OAMP} \cite{yu2022hybrid}, a recently developed NN-based estimator for THz Ultra-Massive MIMO Channel Estimation.
For a comprehensive comparison with the NN-based algorithms, all networks are trained with an 80,000-sample training set and a 5,000-sample validation set. 
We conduct training for 150 epochs utilizing the Adam optimizer, initialized with a learning rate of 0.001 and a batch size of 128.

\begin{table}[t]
\centering
\caption{Performance comparison of different algorithms under various SNR levels}
\label{table:performance}
\begin{tabular}{cccccc}
\toprule
\multirow{2}{*}{Algorithm} & \multicolumn{5}{c}{SNR (dB)}                   \\ \cmidrule(lr){2-6}
                           & 0     & 5     & 10    & 15    & 20    \\ \midrule
LS                      &   0.0021 & -0.89 & -1.28 & -1.42 & -1.46  \\
MMSE                 &     -2.64 & -4.57 & -6.61 & -8.35 & -9.37  \\
OAMP          &             -2.33 & -4.01 & -5.69 & -6.72 & -7.53   \\
FISTA        &              -3.55 & -5.51 & -7.78 & -9.21 & -10.09   \\
ISTA-Net+   &        -7.05 & -10.42 & -13.38 & -15.82 & -17.53 \\
FPN-OAMP     &             \underline{-10.03} & \underline{-13.52} & \underline{-16.93} & \underline{-20.74} & \underline{-23.99}  \\
PR-DEN     &                \textbf{-10.59} & \textbf{-14.18} & \textbf{-18.01} & \textbf{-21.75} & \textbf{-24.75}  \\ \bottomrule
\end{tabular}
\vspace{+.05in}

\footnotesize{
{Notes:The best and second-best scores are \textbf{highlighted} and \underline{underlined}.}
}

\end{table}

Table \ref{table:performance} illustrates the NMSE performance versus SNR of our PR-DEN and six benchmarks in MIMO system with hybrid channels. 
The results clearly illustrate that PR-DEN significantly outperforms LS and MMSE for different SNR values. 
Moreover, in comparison to network-based unfolding algorithms such as ISTA-Net+ and FPN-OAMP, we observe a notable improvement over ISTA-Net+, and approximately a 5\% performance gain compared to FPN-OAMP.
\begin{figure}[b]
    \centering
    \includegraphics[width=0.493\linewidth]{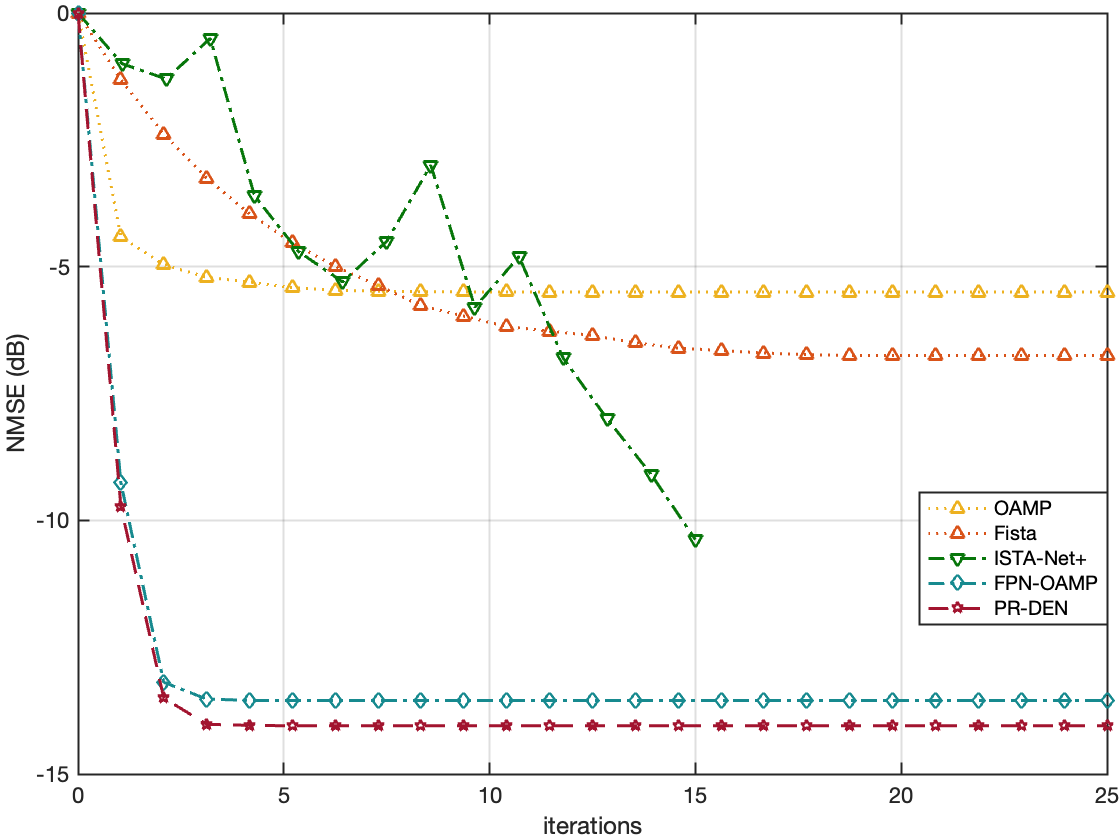}
    \caption{\small{NMSE versus iterations. The curves represent: OAMP (Yellow), Fista (Orange), ISTA-Net+ (Green), FPN-OAMP (Blue), and PR-DEN (Dark Red).}} 
    \label{fig_iter}
\end{figure}
\begin{table}[b]
\centering
\caption{Running time of different algorithms}
\label{table_time}
\resizebox{\columnwidth}{!}{
\begin{tabular}{c|c|c|c|c|c}
\hline
& OAMP & FISTA & ISTA-Net+ & FPN-OAMP & PR-DEN \\
\hline
Center Processor & $0.982$ & $0.103$ & $4.92 \times 10^{-2}$ & $6.39 \times 10^{-3}$ & $7.12 \times 10^{-3}$ \\
\hline
Graphics Processor & / & / & $5.52 \times 10^{-3}$ & $3.61 \times 10^{-4}$ & $4.06 \times 10^{-4}$ \\
\hline
\end{tabular}}
\vspace{0.1in}

\footnotesize{
{Notes: a) Computation times (in seconds) are averaged over five instances at different SNR levels (0:5:20 dB).}
{b) A uniform convergence criterion is used across all algorithms: $\|\boldsymbol{x}^{k+1}-\boldsymbol{x}^{k}\| < 10^{-2}$.}
}
\end{table}

Fig. \ref{fig_iter} illustrates the NMSE performance over iteration steps at an SNR of $5$ dB, compared with the benchmarks, excluding LS and MMSE since these two algorithms are not iterative algorithms. 
The results demonstrate that PR-DEN achieves rapid convergence within $4$ iterations. 
Furthermore, when compared to classical algorithms such as OAMP, FISTA, and ISTA-Net, our approach exhibits superior NMSE performance as early as the second iteration. 
Additionally, compared to FPN-OAMP, our approach converges faster and has better performance in terms of the NMSE validation.

Table \ref{table_time} illustrates a comparison of the computational times between the four benchmark iterative algorithms with PR-DEN.
The computational time is calculated respectively, using a center processor (typical examples known as Central Processing Units) and a graphics processor (typical examples known as Graphics Processing Units). 
FISTA and OAMP primarily involve matrix operations and simple thresholding, which are computationally efficient on a center processor, making graphics processor acceleration unnecessary \cite{beck2009fista}, \cite{ma2017orthogonal}.
Compared to traditional iterative algorithms such as OAMP and FISTA, PR-DEN demonstrates significantly lower computational time. 
In comparison with the network unfolding algorithms, our approach exhibits lower execution time than ISTA-Net+ and is comparable to FPN-OAMP. 
These experimental findings validate the efficiency of our algorithm in handling the channel estimation problem.

\section{Conclusion} \label{sec_conclusion}
This paper introduces a novel channel estimation methodology for MIMO system, addressing the dual problem through the Peaceman-Rachford (PR) splitting approach with the deep equilibrium (DEQ) network. 
This approach constructs non-expansive operators using PR splitting on the Fenchel conjugate of the channel estimation problem, forming a fixed-point equation. 
The DEQ model implicitly learns the proximal operators, enabling efficient iterative steps with constant memory. 
Theoretical analysis provides a convergence proof of the sequence produced by the proposed approach. 
The simulations demonstrate the approach's effectiveness, that our approach has high accuracy and efficiency for channel estimation in MIMO system, compared with six benchmarks.


\bibliographystyle{IEEEtran}


\newpage

\end{document}